\documentclass[a4paper]{article}

\usepackage[left=1.5cm]{geometry}
\usepackage{amssymb}
\usepackage{eurosym}
\usepackage{amsfonts}
\usepackage{amsmath}
\usepackage{amsthm}
\usepackage[pdftex]{graphicx}

\usepackage[innercaption]{sidecap}
\usepackage{latexsym}
\usepackage{lscape}

\newtheorem{theorem}{Theorem}

\theoremstyle{remark}
\newtheorem{remark}{Remark}

\theoremstyle{lemma}
\newtheorem{lemma}{Lemma}

\begin{document}

\author{M.-F. Danca\\
Romanian Institute of Science and Technology,\\400487 Cluj-Napoca, Romania\\
danca@rist.ro\\
M. Romera, G. Pastor, F. Montoya\\
Instituto de F\'{\i}sica Aplicada, Consejo Superior de
Investigaciones Cient\'{\i}ficas,\\ Serrano 144, 28006 Madrid,
Spain}

\title{Finding Attractors of Continuous-Time Systems by Parameter Switching}

\maketitle

\begin{abstract}
The review presents a parameter switching algorithm and his
applications which allows numerical approximation of any attractor
of a class of continuous-time dynamical systems depending linearly
on a real parameter. The considered classes of systems are modeled
by a general initial value problem embedding dynamical systems which
are continuous and discontinuous with respect to the state variable,
and of integer and fractional order. The numerous results, presented
in several papers, are systematized here on four representative
known examples representing the four classes. The analytical proof
of the algorithm convergence for the systems belonging to the
continuous class is briefly presented, while for the other
categories of systems the convergence is numerically verified via
computational tools. The utilized numerical tools necessary to apply
the algorithm are contained in five appendices.
\end{abstract}


\section{Introduction}

In Nature there are many different interactions and the real systems
could evolve according to more that one dynamics for short periods
of time. Therefore, it is reasonable to think that the evolution of
some natural processes could be imagined as the result of the
alternation of different dynamics for relatively short periods of
time. In particular, a topic of research regarding parameter
switching which has arisen in the last years, consists in studying
the dynamics of continuous-time systems
\cite{Danca1,Danca2,Danca5,Danca3,DancaW,DancaMorel,Danca4} and
discrete systems \cite{Almeida,Romera,Danca7}.

In this paper we review aspects of some previous results, obtained
by us with parameter switching techniques. We considered general
classes of systems, continuous and discontinuous with respect to the
state variable, and of integer and fractional order.

\noindent The truthfulness of the previous results are sustained
here by another numerical tool, the cross-correlation.

Via numerical simulations we found for a large class of systems (and
analytically proved for a particular class \cite{Yu}), that any
attractor of some considered system can be \emph{synthesized}
(numerically approximated) by using some parameter switching rule.
The analytical and numerical proofs are based on the fact that the
invariant sets obtained with the control parameter periodically
switched are numerical approximations of those corresponding to the
control parameter replaced with the average of the switched values.
This useful result was intensively verified on several examples with
a numerical algorithm called the Parameter Switching algorithm which
represents an elegant and easy way to numerically approximate any
attractor of a dynamical system, belonging to a general class of
systems, starting from a set of accessible parameter values which
are switched in relative short period of times. The switching rule
can be modeled by some piecewise continuous function. The algorithm
is useful for example in practice, when a desired parameter value
cannot be directly accessed. Also, it can help to understand what
happens in some real systems when the control parameter is switched
by natural or imposed causes.

The Parameter Switching algorithm differs from the known
control/anticontrol algorithms, where the parameter is generally
slightly modified following some very precise rules in order to
modify the behavior of some trajectory. Our algorithm allows the
choice of any deterministic or even random switching rule within a
set of parameter values, the result being an attractor which belongs
to the set of all existing attractors of the underlying system.

The paper is organized on two main parts concerning theoretical
aspects and applications respectively, as follows: Section \ref{doi}
presents the attractors synthesis, where the Parameter Switching is
detailed, Section \ref{attr sinth} presents the numerically evidence
of the Parameter Switching algorithm convergence, and finally, in
Section \ref{applic}, four representative examples are analyzed.
Additional information are presented in five Appendices.

\section {Attractors synthesis}\label{doi}
In this section, the used notions, results, assumptions and the
underlying Initial Value Problem (IVP) modeling a general class of
systems, continuous or discontinuous with respect to the state
variable and of integer or fractional order, are presented.

\subsection{Preliminaries notions and notations}

\noindent All the considered systems can be modeled by the following
IVP

\begin{equation}
\begin{tabular}
[c]{c} $\frac{d^{q}x(t)}{dt^{q}}=f(x(t))+pBx(t)+Cs(x(t))$,~
$x(0)=x_{0},~\ ~t\in I$ \label{IVP}
\end{tabular}
\end{equation}

\noindent where $p$ is a real parameter, $q$ stands for the
derivative order (for $q = 1$, we have the known standard
derivative, while for $q \ne 1$ we have the so called
\emph{fractional derivative}: ${{{d^q}} \mathord{\left/
 {\vphantom {{{d^q}} {d{t^q}}}} \right.
 \kern-\nulldelimiterspace} {d{t^q}}}$),
$f: \mathbb{R}^{n} \rightarrow \mathbb{R}^{n} $ is a nonlinear
vector valued function, at least continuous with respect to the
state variable, $I = \left[ {0,T} \right],\,\,T
> 0$ , \emph{B}, \emph{C} real $n\times n$ squared matrices, and
$ s:\mathbb{R}^{n}\rightarrow \mathbb{R}^{n}, \;
s(x)=(s_1(x_1),\ldots, s_n(x_N))^t$ is a piece-wise continuous
function, being composed in most general cases by signum functions:
$s_i(x_i) = sgn({x_i}), \; i=1,2,\ldots,n$, or e.g. step (Heaviside)
functions.

\noindent It is classically assumed that $q \in \left( {0,1}
\right]$. Function of \emph{q} and \emph{C}, we can have the
situations presented in Table \ref{tab1}.

\noindent Throughout this paper the following assumption will be
considered

\mbox{}

\noindent (\textbf{H1}) {\it The IVP (\ref{IVP}) admits a unique
solution (e.g. Lipschitz continuous)}.

\mbox{}

The control parameter $p$ is considered to be a (periodic) piecewise
constant function $p:I\rightarrow \mathbb{R}$ (an example for a
periodic function $p$ is presented in Fig.\ref{fig0}). As we shall
see next, $p$ can be a periodic or non-periodic function. Other form
of functions for $p$ can be found in \cite{Yu}.

Due to the piece-wise continuity of $p$, the IVP (\ref{IVP}) becomes
non-autonomous. However, for the sake of simplicity, next the time
variable \emph{t} will be omitted unless necessary. Therefore, the
IVP (\ref{IVP}) can be written as follows

\begin{equation}
\frac{{{d^q}x}}{{d{t^q}}} = f(x) + pBx + Cs(x),\,\,x(0) =
{x_0},\,\,t \in I.\label{IVPsimplu}
\end{equation}

\begin{remark}
\noindent The existence and uniqueness conditions for IVPs modeling
\emph{DI} and \emph{DF} systems differ from those for \emph{CI}
systems, and are not presented here (for our class of \emph{DI}
systems they can be found in e.g. \cite{Lempio2}, while for
differential equations of fractional order in \cite{Diethelem2} or
\cite{kaicarte}.)
\end{remark}

The systems chosen to represent in this work the four classes in
Table \ref{tab1} are three-dimensional, but the algorithm and the
underlying results are applicable for any finite lower or higher
dimension $n$.

\noindent The examples treated here are presented in Table
\ref{tab2}: Lorenz system, Sprott system \cite{Sprott}, L\"{u}
system \cite{Lu2} and a fractional variant of the Chua's system
\cite{Brown}. Other examples can be found in
\cite{Danca1,Danca5,DancaW,DancaMorel,Danca4}.

A \emph{global attractor}, roughly speaking, is viewed in this paper
as a state space region of a dynamical system where the system can
enter but not leave, and containing no smaller regions (see e.g.
\cite{kapitan}). The global attractor contains all the dynamics
evolving from all initial conditions. In other words, it contains
all the solutions, including stationary solutions, periodic
solutions, as well as chaotic ones. The term of \emph{local
attractor} is used sometimes for attractors which are not global
attractors \cite{Stuart}.

\noindent The global attractors may contain several local
attractors. Therefore, a global attractor can be considered as being
``composed'' of the set of all local attractors for a given
parameter $p$ value and initial conditions. Each local attractor
attracts trajectories from a subset (basin of attraction) of initial
conditions (for details on the notions of local and global
attractors we refer e.g. to \cite{Milnor,Temam,Hirsch,Hirsch2}).

\begin{remark}
(i) For the sake of simplicity, when a global attractor is composed
by several local attractors, only a single local attractor will be
considered (the choice can be made by appropriate selections of the
initial conditions). Therefore, hereafter, by \emph{attractor} one
understands, simply, either one of the local attractors or the
single local attractor which composes the global attractor;

\noindent(ii) Due to the predominant numerical characteristics of
the present work, without a significant loss of generality, the
attractors will be considered as approximations, after neglecting a
sufficiently long period of transients \cite{Foias}, of the $\omega
-limit$ sets (the set of points that can be limit of
subtrajectories). Despite the fact that usually these sets are
uncomputable, they can be numerically approximated. Therefore, in
this paper the attractors are considered as being the plots of the
$\omega -limit$ sets.
\end{remark}

\noindent Let us use throughout the review the following notations

\mbox{}

 \noindent\textbf{Notation 1}\label{notatie}
\mbox{} \noindent- $\mathcal{A}$ the set containing the attractors
depending on $p$, including attractive stable fixed points, limit
cycles and chaotic attractors;

\noindent-  $\mathcal{P}$ the set of all \emph{p} admissible values;

\noindent-  $\mathcal{P_{N}}=\{p_{1},p_{2},\ldots,p_{N}\}\subset
\mathcal{P} $ a finite ordered subset of $\mathcal{P}$;

\noindent- $\mathcal{A}_{N}=\{A_{1},A_{2},\ldots,A_{N}\}\subset
\mathcal{A}$ the set of the attractors corresponding to
$\mathcal{P_{N}};$

\noindent- $I = \bigcup\limits_{j = 1,2, \ldots }
({\bigcup\limits_{i = 1}^N {I_{ij}} })$, where the adjoint
subintervals $I_{ij}$ are of time length $m_{i}h$, where the
"weights" $m_i$ are some positive integers, $h>0$, for $i=1,2,
\ldots, N $ and all \emph{j} (see Fig.\ref{fig1} for the particular
case of the first set of time-intervals $I_{i1}$ for $i=1,2,3,4$ );

\noindent- $p^*$ the \emph{average parameter}

\vspace{1mm}
\begin{equation}\label{averaged}
{p^*} = \frac{{\sum\limits_{i = 1}^N {{m_i}{p_i}} }}{{\sum\limits_{i
= 1}^N {{m_i}} }};
\end{equation}
\vspace{1mm}

\noindent- $A^*$ the $\emph{average attractor}$, obtained for
$p=p^*$.

\begin{remark} \label{bijectie}
\noindent Taking account to the Assumption \textbf{H1} it follows
naturally to define a monotone bijection
$F:\mathcal{P}_{N}\rightarrow\mathcal{A}_{N}$ for some fixed
\emph{N}. Therefore, to each $p \in \mathcal{P_N}$ corresponds a
unique element $A \in \mathcal{A_N}$ and reversely, for each $A \in
\mathcal{A_N}$ there exists $p \in \mathcal{P_N}$ such that $A =
F(p)$ (Fig.\ref{fig2}).
\end{remark}

\subsection{Parameter Switching algorithm}\label{algorithm}

To prove that any attractor can be approximated by switching the
parameter while the underlying IVP is integrated, we need a
numerical algorithm to implement the switches that we name Parameter
Switching (PS) algorithm.

\noindent Let fix for some \emph{N}, the set $\mathcal{P}_{N}$.
Then, ${p^*}$ given by (\ref{averaged}), can be rewritten in the
following form
\begin {equation}\label{avconvex}
 {p^*} =\sum\limits_{i = 1}^N {{\alpha _i}{p_i}} \emph{  with  } {\alpha _i} = {{{m_i}} \mathord{\left/
 {\vphantom {{{m_i}} {\sum\limits_{i = 1}^N {{m_i}} }}} \right.
 \kern-\nulldelimiterspace} {\sum\limits_{i = 1}^N {{m_i}} }},\,\,{p_i} \in
 {P_N}.
\end {equation}

\noindent Because ${\alpha _i} < 1$ and $\sum\limits_{i = 1}^N
{{\alpha _i} = 1} $ , ${p^*}$ enjoys the following property

\mbox{}

\noindent {\textbf{P1.}} \label{convex} {\it For every set
$\mathcal{P}_{N}$, $p^*$  given by (\ref{averaged}) is a convex
combination of $p_i$, $i=1,2,\ldots,N$.}

\mbox{}

To implement the \emph{PS} algorithm, we have to integrate the IVP
(\ref{IVPsimplu}) with a numerical scheme for ODEs with single
step-size {\emph{h}}
 (e.g. the standard
Runge-Kutta method).

\noindent Let first consider that $p$ is a periodic function of
period $T_0$, i.e. $p(t+T_0)=p(t)$ for all $t$ in $I$. While the
solution to the IVP (\ref{IVPsimplu}) is numerically approximated,
the parameter \emph{p} is periodically switched within
$\mathcal{P}_{N}$ in every consecutive time interval $I_{ij}$,
following some designed scheme, denoted hereafter by $S_h$

\begin{equation}\label{S}
    S_{h}\circeq \left[ {{p_1}\left| {_{{I_{1j}}},} \right.{p_2}\left| {_{{I_{2j}}}, \ldots ,{p_N}\left| {_{{I_{Nj}}}} \right.} \right.} \right],\,\,j = 1,2,
    \ldots,
\end{equation}
\mbox{}

\noindent which means that while the IVP (\ref{IVPsimplu}) is
integrated, in each interval ${I_{ij}}$, \emph{p} will be replaced
by $p_{i}$ for every $j=1,2,\ldots $. Thus, for $t\in I_{11}$ ,
$p(t)=p_1$, for $t\in I_{21}$, $p(t)=p_2$ and so on until $I_{N1}$,
when $p(t)=p_N$. On the next interval $I_{12}$, again $p(t)=p_1$ and
so on until the interval $I_{N2}$, when $p(t)=p_N$. The algorithm
repeats on the next set of intervals $I_{i3}, i=1,2,\ldots,N$ and so
on periodically, until $t\geq T$. In other words, $p$ is a piecewise
constant and periodic function of period
$T_0=h\sum\limits_{i=1}^{N}m_{i}$ having the following expression
(See Fig.\ref{fig0},\ref{fig1})
\begin{equation}\label{p0}
p(t) = {p_i},\,\,for\,\,t \in {I_{ij}},\,\,\,i = 1, \ldots
,N,\,\,\,j = 1,2, \ldots
\end{equation}

\noindent The length of the time intervals $I_{ij}$ will be taken as
multiple of $\emph{h}$: $length(I_{ij})=m_{i}h$ for each $\emph{j}$.
Therefore, for a fixed {\emph{h}},  $S_h$ can be noted in a
simplified form
\begin{equation}\label{Ssimplu}
    S_{h}\circeq \left[ {{m_1}{p_1},{m_2}{p_2}, \ldots ,{m_N}{p_N}} \right],
\end{equation}

\noindent which means the following $\emph{p}$ infinite sequence

\begin{center}
\noindent ${m_1}{p_1},{m_2}{p_2}, \ldots
,{m_N}{p_N},{m_1}{p_1},{m_2}{p_2}, \ldots ,{m_N}{p_N}, \ldots $
\end{center}

\noindent For example, $S_h=\left[ {2{p_1},{p_2}} \right]$ for a
given $h$, means that for the time-interval of length $2h$, $p=p_1$
then for the next time-interval of length $h$, $p=p_2$ . Next, for
two integration steps, $p=p_1$, then for one integration step,
$p=p_2$ and so on

\noindent (i.e. periodically with period $T_0=(m_1+m_2)h=3h$).
Schematically, $\emph{S}_h$ can be written as the infinite sequence
$\left[ {2{p_1},{p_2}} \right] =
{p_1},{p_1},{p_2},{p_1},{p_1},{p_2}, \ldots $

\noindent The pseudocode of the $\emph{PS}$ algorithm is given in
Table \ref{tab3}.

\noindent It is easy to verify that
\[
p^\ast =\frac{1}{T_{0}}\int\limits_{t}^{t+T_{0}}p(\tau )d\tau
,~~t\in I.
\]

\noindent \textbf{Notation 2} Let denote by $A^\circ$ the attractor,
obtained with the \emph{PS} algorithm, called hereafter the
\emph{synthesized attractor}.

\begin{remark}\label{random}
\noindent It is easy to see that, for some given $p$, the relation
(\ref{averaged}) considered as equation, may have several solutions.
For example, if we set $N=2$, and want to obtain $p^*=4$ using the
scheme $S_h=[m_1p_1,m_2p_2]$, for $p_1=2$ and $p_2=6$, we can choose
$m_1=m_2=1$ but also $m_1=m_2=3$ to verify (\ref{averaged}). If we
fix $m_1=3$ and $m_2=1$, in order to obtain $p^*=4$, we can use
$p_1=2$ and $p_2=10$, but also $p_1=p_2=4$.
\end{remark}

\section{Numerical proof of PS algorithm convergence}\label{attr sinth}

In this section we prove numerically that for a chosen set of
attractors $\mathcal{A}_{N}$, the synthesized attractor $A^\circ $
obtained with the \emph{PS} algorithm  belongs to $\mathcal{A}_{N}$
and, moreover, $A^\circ$ is approximatively identical to $A^*$.

\noindent In order to compare two attractors, we have to provide the
following criterion

\mbox{}

\noindent $\textbf{Criterion}$ We shall say that two attractors are
\emph{approximatively identical}  ($\emph{AI}$) if their
trajectories in the phase space approximatively coincide, and the
Hausdorff distance (Appendix \ref{D}) is small enough.

\mbox{}

\noindent In our numerical experiments Hausdorff distance was of
order of $10^{-4}-10^{-3}$.

\noindent Due to the bijectivity of \emph{F}, considering the total
order over the set $\mathcal{P}_{N}$, it is reasonable to consider
that the following property holds

\mbox{}

\noindent \textbf{P2.} $\mathcal{A}_{N}$ {\it is an ordered set
endowed with the $\mathcal{P}_{N}$ order induced by the function} $\emph{F}$.\\

\noindent Moreover, the same order can be found over the sets
$\mathcal{P}_{N}$ and $\mathcal{A}_{N}$ considered as intervals:
$\left[ {{p_1},{p_N}} \right]$ and $\left[ {{A_1},{A_N}} \right]$
respectively and, without losing generality, we can consider that
${A_i} = F({p_i}),\,\,i = 1,2, \ldots ,N$ (Fig.\ref{fig3}). This
property is outlined in all bifurcation diagrams.

\mbox{}

 \noindent \textbf{Notation 3} Let denote $\emph{AI}$ by
$"\cong"$.

\mbox{}

\noindent Next, in order to prepare the proof of the main result
regarding the parameter switching, we introduce the following lemma

\begin{lemma}\label{lemma}
Given N and $\mathcal{P_N}$, $A^\circ \cong A^*$.
\end{lemma}

\begin{proof}
The lemma has been checked numerically with tools such as:
histograms, Poincar\'{e} sections, time series, cross-correlation
(Appendix \ref{C}) and Hausdorff distance (Appendix \ref{D}). The
numerous examples, show that the attractor $A^\circ$, obtained with
\emph{PS} algorithm and $A^*$ obtained for $p=p^*$, are \emph{AI},
the degree of the identity depending less or more on the system
characteristics and, unavoidably, on the numerical errors. Hausdorff
distance, for all considered systems, was of order of
$10^{-4}-10^{-3}$.
\end{proof}

\noindent The sketch of the analytical proof of this lemma,
presented in \cite{Yu} for the case of $CI$ systems, can be found in
Appendix \ref{A}.

\begin{remark}
\noindent Applying the symbolic computation for several examples of
\emph{CI} systems, with the scheme $S_h$ for $N\leq3$, the IVP
(\ref{IVPsimplu}) was integrated with the forward Euler method. The
result shown that the (Euclidean) difference between the two
solutions corresponding to $p=p^*$ and to $p$ switched with the
\emph{PS} algorithm, is of order of $O(h^2)$, the same as the error
of the considered Euler method.
\end{remark}

\noindent Yet, the main result which can be numerically proved, can
be introduced.

\begin{theorem}\label{prop}
Given N and $\mathcal{P_N}$, $A^\circ$ belongs to $(A_1,A_N)$.
\end{theorem}

\begin{proof}

\noindent By the properties \textbf{P1} and \textbf{P2}, it follows
that $A^*\in(A_1,A_N)$. Next, by the Lemma \ref{lemma}, the
attractor $A^\circ$ synthesized with some scheme $S_h$, is \emph{AI}
to $A^*$. Thus, $A^\circ \cong A^*$ and therefore $A^\circ$ belongs
to $(A_1,A_N)$, which completes the proof (see Fig.\ref{fig00}).
\end{proof}

\noindent Summarizing, for every finite set $\mathcal{P_N}$ and
numbers $m_i$, the synthesized attractor $A^\circ$ will belong to
$(A_1,A_N)$, and differs from every attractor $A_i\in \mathcal{A_N},
\ \ {\it i}=1,2,\ldots N$ (due to the convexity property).
Reversely, any attractor of $\mathcal{A_N}$ can be considered as
being synthesized with the $\emph{PS}$ algorithm, by means of a
finite set of attractors of $\mathcal{A_N}$.

For the continuous case, the analytical proof in \cite{Yu}, shows
that the solutions of the equation (\ref{IVPsimplu}) with $p$
switched within $\mathcal{P_N}$ with \emph{PS} algorithm and that
with $p$ replaced with $p^*$ can be arbitrarily close. Therefore,
the underlying invariant sets (attractors in our case) are also
arbitrarily (\emph{AI}) close (\cite{Stuart}, Ch. 6).

Due to the mentioned convexity property, whatever kind of
combinations of $p_i$ and $m_i$ values are considered for a fixed
$N$, $p^*$ will still belong to the interval $(p_1,p_N)$. Therefore,
for $p$ we can chose any kind of piecewise continuous functions,
with the only condition that their values and $p^*$ belongs within
$\mathcal{P_N}$ (see examples in \cite{Yu}).

The proof of the convergence (analytically or numerically verified)
does not depends on the periodicity of $p$ but only on the convexity
of $p^*$. Therefore, it is obvious that not only periodic schemes
(\ref{Ssimplu}) can be used, but even random ones \cite{Danca1}. One
of the simplest way to implement randomly the \emph{PS} algorithm,
once $N$ is fixed, is to chose first $p_i$ and $m_i$ in some random
manner, after which the \emph{PS} algorithm is started (see the
example of Sprott system, Subsection \ref{DI}). Obviously, there are
several other random ways such as: choosing randomly $m_i$ and $p_i$
while the \emph{PS} is running, or switching the order of $p_i$ and
so on. Now, the averaged $p^*$ has to be determined with the
following relation
\begin{equation}\label{p random}
   {p^*} = \frac{{\sum\limits_{i = 1}^N {m_i^,{p_i}} }}{{\sum\limits_{i = 1}^N {m_i^,} }},\
\end{equation}

\noindent where $m^{'}_i$ denote the number of times when $p_i$ is
chosen by the algorithm for $t\in I$.

\begin{remark} \label{remarca cu erori}

\noindent \label{precizie} (i) The ``structural stability'' of the
$PS$ algorithm presents some obvious limitations due firstly to his
numerical approach (some details and other related aspects can be
found in \cite{Danca1}). For example, for relative large values of
$m_i$, the trajectory of $A^\circ$ could present some ``corners'' (a
maximum difference between $m_i$ should generally be about $(20\div
25)h$). The values for $p_i$ could be taken over the entire set
$\mathcal{P_N}$ without distinguishable differences between
$A^\circ$ and $A^*$. Excessive number of decimals for $p^*$ could
lead too to some differences between the two attractors $A^\circ$
and $A^*$. Even for large values for \emph{N}, $A^\circ$ and $A^*$
still remain close each other;

\noindent (ii) The cross-correlation and time series show an
interesting characteristic: the trajectories corresponding to
$A^\circ$ and $A^*$, even in the phase space and time
representations are $\emph{AI}$, they are dephased in time (see
cross-correlation in the figures);
\noindent \label{infinit} (iii) For chaotic attractors, the
\emph{AI} is obtained only ``asymptotically'' since the necessary
time to fully approximate the attractor is, theoretically, infinite.
\end{remark}

\noindent The $\emph{PS}$ algorithm can be used to ``control'' or
``anticontrol'' dynamical systems modeled by the IVP (\ref{IVPsimplu})
when some targeted parameter value cannot be accessed directly (see
\cite{Danca3}). For this purpose, we have to choose $p_i$, $m_i$ and
some scheme $S_h$ to obtain the targeted value $p^*$. However, while
almost all known control/anticontrol algorithms ``force'' some
trajectory to change its characteristics and behavior, the $PS$
algorithm allows to obtain any desired already existing attractor of
$\mathcal{A_N}$.

\section{Applications}\label{applic}

This section is devoted to the applications of the \emph{PS}
algorithm to the four classes of dynamical systems (see Tables
\ref{tab1} and \ref{tab2}) to synthesize attractors. In this purpose
we have to choose $N, \mathcal{P_N}$ and $S_h$ for each system, such
that a desired value $p^*$ (which can be taken e.g. from the
bifurcation diagram) is obtained.

To apply the \emph{PS} algorithm for \emph{CI} systems, we used the
standard Runge-Kutta method (with the step size $h$ of order between
$10^{-4}$ and $10^{-2}$, depending on the characteristics of the
considered system), while for the discontinuous and fractional
systems, we have chosen special numerical methods. The bifurcation
diagrams, time series, histograms and cross-correlations were
determined and plotted superimposed for the first state variable
$x_1$. The Poincar\'{e} sections have been determined for the plane
$x_3=const$. Some bifurcation diagrams, like the one for the Sprott
system and especially for the fractional L\"{u} and Chua systems,
require an extremely long computer time (see Appendix \ref{B} ). For
discontinuous systems (of integer and fractional order), some
'corners' can be remarked, typical for these kind of systems (the
solution for the underlying IVP are generally not smooth \cite
{Lempio2,Lempio}). As stated before, the \emph{AI} was verified for
all the considered systems via superimposed phase portraits,
Poincar\'{e} sections, histograms, time series and also with
cross-correlation and Hausdorff distance. For all the considered
cases, the results lead to the same conclusion: Lemma \ref{prop}
applies to all considered classes of systems.

\subsection{Continuous dynamical systems of integer order} The
$\emph{PS}$ algorithm was tested on several examples of \emph{CI}
systems such as: Lorenz, Chen, R\"{o}ssler, Rabinovitch-Fabrikant
\cite{Danca1}, Hindmarsh-Rose neuronal system \cite{DancaW},
networks \cite{DancaMorel} and Lotka-Volterra \cite{Danca4}. Here,
we consider the representative case of the Lorenz system.

\noindent For this class of systems, the IVP (\ref{IVPsimplu}) has
to be considered for the particular case $q=1$ and $C=O_{n\times
n}$, namely (Table \ref{tab2})
\begin{equation}\label{Lorenzeq}
\dot x = f(x) + pBx,\,\,x(0) = {x_0},\,\,t \in I.
\end{equation}

\noindent The used numerical method is the standard Runge-Kutta with
integration step-size $h=0.01$.

\begin{itemize}
\item Let first consider the scheme (\ref{Ssimplu}) for $N=2$: $[{m_1}{p_1},$ ${m_2}{p_2}]$ with $p_1=90$, $p_2=96$, and $m_1=m_2=1$. Then $p^*$, given by the relation (\ref{averaged}), is $p^*=(m_1\times p_1+m_2\times p_2)/(m_1+m_2)=93$. Applying the $\emph{PS}$ algorithm, the synthesized attractor $A^\circ$ (red plot, Fig.\ref{fig4} a) is a stable limit cycle that is $\emph{AI}$ with the average attractor $A^*$ (superimposed blue plot over $A^\circ$) for $p^*=93$. The $\emph{AI}$ is emphasized in addition to the over-plot in the phase space, by the superimposed Poincar\'{e} sections with the plane $x_3=130$ (Fig.\ref{fig4}b) and superimposed histograms for the first state variable $x_1$ too (Fig.\ref{fig4}c ). The cross-correlation (Fig.\ref{fig4} d) shows that the time series corresponding to $A^\circ$ and $A^*$ are \emph{AI}, but dephased. This time-difference between the corresponding trajectories is better remarked from the time series corresponding to the first  component $x_1$ (see Fig.\ref{fig4} e).
\item Let next consider the case $N=5$ with the scheme $\left[{2{p_1},3{p_2},2{p_3},4{p_4},3{p_5}} \right]$ for $p_1=125$, $p_2=130$, $p_3=140$, $p_4=144$ and $p_5=220$. In order to facilitate the use of the \emph{PS} algorithm, the bifurcation diagram will be used (Fig.\ref{fig5}). Now, $p^*=154$, and the synthesized attractor $A^\circ$ is again a stable limit cycle which is \emph{AI} with $A^*$ (Fig.\ref{fig6} f) even $A_{1-4}$ are chaotic and only $A_5$ is a stable limit cycle (Fig.\ref{fig6} a-e). Both attractors $A^\circ$ and $A^*$ are $\emph{AI}$ (see Fig.\ref{fig6} g-h wherefrom the \emph{AI} property can be remarked). The time series being dephased (Fig.\ref{fig6} i,j), the trajectories of the attractors $A^\circ$ and $A^*$ are \emph{AI}, but time dephased.
\item If for the same scheme  $\left[{2{p_1},3{p_2},2{p_3},4{p_4},3{p_5}} \right]$ we choose $p_5=166$ instead $p_5=220$, the synthesized attractor $A^\circ$ is chaotic and \emph{AI} with $A^*$ for $p^*=142.428$ (Fig.\ref{fig7}). However, now the $\emph{AI}$ is only an $\emph{almost}$ identity (see Remark \ref{remarca cu erori} (iii)). The Poincar\'{e} section (Fig.\ref{fig7} c) was obtained with the plane $x_3=145$. The cross-correlation shows that the underlying trajectories of $A^\circ$ and $A^*$ are dephased. Because the trajectories are chaotic, the time series to underline this time-difference is irrelevant in this case.
\end{itemize}

\noindent For all analyzed examples, the Hausdorff distance was of
order $10^{-3}$.

\noindent Other examples of \emph{CI} systems can be found in \cite
{Yu}.

\subsection{Continuous dynamical systems of fractional
order}\label{CF}

Fractional mathematical concepts allow to describe certain real
objects more accurately than the classical ``integer'' methods.
Examples of such real objects that can be elegantly described with
the help of fractional derivatives displaying fractional-order
dynamics, may be found in many fields of science and exhibit a wide
range of rich dynamics. Therefore, the fractional calculus starts to
attract increasing attention of mathematicians but also of
physicists and engineers (see e.g. \cite{Polod,Hilfer,Ahmed,el}).

Many \emph{CF} systems, can be modeled by the IVP (\ref{IVPsimplu})
with $q\in (0,1)$ and $C=O_{n\times n}$. The fractional derivative
$\frac{{{d^q}}}{{d{t^q}}}$ is generally denoted using the Caputo
differential operator of order {\emph{q}}, $D^q_*$ (see e.g.
\cite{Podlubny}). Thus, the IVP (\ref{IVPsimplu}) becomes

\begin{equation}
\label{IVPfrac} D^{q}_*x =f(x)+pBx, \quad  x^{(k)}(0) =
x_{0}^{(k)},~ (k = 0,1,\ldots, \lceil q \rceil-1 ).
\end{equation}

\noindent $\lceil .\rceil$  denotes the ceiling function that rounds
up to the next integer, and $D^m_*=\frac{d^m}{d{t^m}}\
 $, with $m=\lceil q \rceil$, is the standard
differential operator of the integer order $\lceil q \rceil\in
\mathbb{N}$. The Caputo operator, with starting point 0, has the
following expression
\begin{equation*}\label{Caputo}
D^q_* x(t) = \frac 1 {\Gamma(m - q)} \int_0^t (t - \tau)^{m - q - 1}
D^m_* x(\tau) {\mathrm d} \tau.
\end{equation*}

\noindent where $\Gamma$ is the Gamma function (Appendix \ref{B}).
Because $D^q_*$ has an \emph{m}-dimensional kernel, \emph{m} initial
conditions need to be specified. Therefore, for the common case
chosen in this paper $0 < q<1$, we have to specify just one
condition, in the classical form \cite{Diethlem et al}: $x^{\left( 0
\right)}(0)=x_0$.

To implement the $\emph{PS}$ algorithm in this case, it is necessary
to choose a numerical method for the solution to the IVP
(\ref{IVPfrac}). In this purpose we use the fractional Adams-
Bashforth- Moulton method (see Appendix \ref{B}) introduced in
\cite{Diethlem et al}.

\noindent Let choose for our purpose the fractional variant of the
L\"{u} system (see Table \ref{tab2}) which unifies the Lorenz and
Chen systems, presented by L\"{u} {\it et al.} in \cite{Lu}. As many
of the real fractional systems have the order of the fractional
differential operators less than 1, we fix in this paper $q=0.9$
(see \cite{Danca6}) which is a typical value exhibiting all the
relevant phenomena (the dynamics of this system, as \emph{q} varies,
can be found in \cite{Lu2}, while some aspects of the attractors
synthesis of the fractional L\"{u} system is presented in \cite{cu
Dieth}).

\noindent Now, the IVP (\ref{IVPfrac}) was integrated with the
fractional Adams-Bashforth-Moulton method with step size $h=0.005$
and $15000\div 20000$ steps, in function of the dynamics of the
synthesized attractor $A^\circ$.

\begin{itemize}
\item A chaotic attractor $A^\circ$ can be obtained with the scheme $[1p_1,1p_2]$ (see Fig.\ref{fig8} c) with $p_1=32$ and $p_2=34.5$. The attractors corresponding to $p_1$ and $p_2$ are plotted in Fig.\ref{fig8} a, b which, as can be seen in the bifurcation diagram in Fig.\ref{fig9}, are stable limit cycles. $p^*=33.25$, and due to the chaotic behavior, the \emph{AI} between $A^\circ$ and $A^*$ is only asymptotic (see the histograms in Fig.\ref{fig8} d and the Poincar\'{e} sections in Fig.\ref{fig8} f). Even $A^\circ$ and $A^*$ are chaotic, from the cross-correlation (Fig.\ref{fig8} e), one can see  that they are dephased, but still \emph{AI}.
\item If we choose the scheme $[1p_1,1p_2]$, with $p_1=33.5$ and $p_2=35.5$, a stable limit cycle $A^\circ$, for which $p^*=34.5$, is obtained (Fig.\ref{fig10}). The \emph{AI} can be remarked from the Poincar\'{e} section and histograms (Fig.\ref{fig10} b,c). Again, the time difference between the underlying time series can be seen from the cross-correlations (Fig.\ref{fig10} d) and time series (Fig.\ref{fig10} e)
\end{itemize}

\noindent For all tested \emph{CF} systems, the Hausdorff distance
was only of order of $10^{-2}$, compared e.g. with $CI$ systems,
where it was of order of $10^{-3}$. This is explainable due to the
well known $O(h^2)$ error bound for the used one-step
Adams-Bashforth-Moulton method for fractional systems (detailed
discussions on errors can be found in \cite{Diethelem2}).

\subsection{Discontinuous dynamical systems of
integer order}\label{DI}

Differential equations with discontinuous right-hand side, model a
whole variety of realistic applications: dry friction, electrical
circuits, oscillations in visco-elasticity, brake processes with
locking phase, oscillating systems with viscous dumping,
electro-plasticity, convex optimization, control synthesis of
uncertain systems and so on (see e.g. \cite{Popp,Deimling,Wierci}
and the references therein).

For our class of \emph{DI} systems, $q=1$ and $C\neq O_{n\times n}$
and the IVP (\ref{IVPsimplu}) becomes

\begin{equation}\label{IVPdisc}
    \dot x = f(x) + pBx + Cs(x),\,\,\,\,x(0) = {x_0},\,\,\,t \in I.\
\end{equation}

\noindent In this case, the right-hand side is discontinuous for a
null set of points $M$ where $s$ vanishes, and continuous in
$D=\mathbb{R}^n\setminus M$\footnote{In \cite{Danca2007} a
classification of systems modeled by the IVP (\ref{IVPdisc}) is
presented.}. Obviously, the IVP (\ref{IVPdisc}) cannot be solved
with classical methods. For example, for the equation
\begin{equation}\label{exdisc}
\dot x = 2 - 3{\mathop{\rm sgn}} (x),
\end{equation}

\noindent where $M = R\backslash ({D_1}\bigcup {{D_2}) = \{ 0\}}$
with $D_1=(-\infty,0)$, $D_2=(0,\infty)$, the classical solutions,
for $x\neq 0$, are
\begin{equation}
x(t) = \left\{ \begin{array}{l}
 5t + {C_1},\,\,\,x < 0, \\
  - t + {C_2},\,\,x > 0, \\
 \end{array} \right.
\end{equation}

\noindent with the integration constants $C_1, C_2$ but, as $t$
increases, these solutions tend to the line $x = 0$, where they
cannot continue to evolve along this line since the function $x(t) =
0$ does not satisfy the equation (Fig.\ref{fig11}). Thus, there is
no classical solution starting from $0$.

Therefore, the problem has to be restarted as a differential
inclusion by using, for example, the Filippov regularization
(Appendix \ref {filippov}). Thus, the IVP (\ref{IVPdisc}) transforms
into a differential inclusion (set-valued IVP)

\begin{eqnarray}
\label{IVPdi} \overset{.}{x}\in f(x)+pBx+CS(x),~~~x(0)=x_{0},~
\text{\emph{for~~almost~all}}~~t\in I, \nonumber
\end{eqnarray}

\noindent where $S$ is the setvalued variant of $s$. On mild
assumptions, a differential inclusion has a solution that happens to
be even unique, but it could have multiple solutions too. To find
them numerically, in our particular case of the IVP (\ref{IVPdisc}),
we can use the standard Runge-Kutta method within $D$ and a special
numerical method for differential inclusions in $M$ (the simplest
forward Euler method here, see Appendix \ref{filippov}).

\noindent Once we set the numerical method for the IVP
({\ref{IVPdisc}), we can apply the \emph{PS} algorithm. For this
class of \emph{DI} systems, we choose the Sprott system
\cite{Sprott} (Table \ref{tab2}).

\begin{itemize}
\item First, let us choose $N=2$ and the scheme $[1p_1,1p_2]$ for $p_1=0.5$ and $p_2=0.528$, for which the corresponding attractors $A_1$ and $A_2$ are chaotic (Fig.\ref{fig12}). We have chosen this scheme such that the obtained average value $p^*=0.514$ belongs to a stable periodic window in the bifurcation diagram. Therefore, the synthesized attractor is a stable limit cycle (Fig.\ref{fig13} c). The \emph{AI} is underlined by the superimposed Poincar\'{e} sections (with the plane $x_3=0$, Fig.\ref{fig13} d) and histograms (Fig.\ref{fig13} e). The time-difference between the trajectories is remarked from the cross-correlation (Fig.\ref{fig13} f) and time series (Fig.\ref{fig13} g).

\item As seen in Section 3, \emph{PS} algorithm can be implement in random manners too. For example, for $N=100$ if one choose randomly  $m_i \in \{ 1,2,3 \}$ and $p_i \in [0.45,0.65]$ for $i=1,2,\ldots,100$ with uniform distribution, and with the obtained values we launch \emph{PS}, the synthesized attractor $A^\circ$ is still \emph{AI} to the average attractor $A^*$ (Fig.\ref{fig14}). However, taking into account the asymptotic generation of chaotic attractors, and the relative large value for $N$, the small differences between the two attractors, $A^\circ$ and $A^*$ are explainable.
\end{itemize}

\subsection{Discontinuous dynamical systems of
fractional order} \label{DF}

There are real discontinuous dynamical systems which display
fractional-order dynamics. We consider here the following class of
\emph{DF} systems, modeled by the IVP (\ref{IVPsimplu}) for $C\neq
O_{n\times n}$ and $p<1$
\begin{equation}\label{fracdisc}
\frac{{{d^q}x}}{{d{t^q}}}= f(x) + pBx + Cs(x),~ x(0) = {x_0},~ t \in
I.
\end{equation}

\noindent In \cite{Dancax} it is proven that the IVP admits
solutions which can be numerically determined.

\noindent Shortly, the IVP is transformed first into a differential
inclusion via the Filippov regularization (as in Subsection
\ref{DI}). Next, using the Cellina's theorem (see e.g. \cite[p.
84]{Aubin1}) the set-valued IVP of fractional-order is transformed
into a continuous single-valued of fractional-order IVP (see for
continuous approximation of \emph{DI} systems the way chosen in
\cite{DancaCodr}). The approximation is made in a sufficiently small
$\varepsilon$-neighborhood of the discontinuity points. To be
precise, let us consider the simplest example of the scalar
function, widely used in examples: $s (x) = c\,sgn(x)$. To
approximate $s(x)$ in an $\varepsilon$-neighborhood of $x=0$, we can
choose one of the simplest function, the \emph{sigmoid}
$h_\varepsilon$
\begin{equation}
h_{\varepsilon }\left( x\right) =c
\bigg(\frac{2}{1+e^{-x/\varepsilon }}-1\bigg).
\end{equation}

\noindent For our general case of the IVP (\ref{fracdisc}), the
continuous approximation leads to the following continuous IVP of
fractional-order
\begin{equation}
\frac{{{d^q}x}}{{d{t^q}}} - f(x) - pBx = \left\{
{\begin{array}{*{20}{c}}
   {Cs(x),\,\,\,for\,\,\,x \notin M,}  \\
   {{h_\varepsilon }(x),\,\,\,\,\,for\,\,\,x \in M,\,}  \\
\end{array}} \right.
\end{equation}

\noindent where $h_\varepsilon(x)$ is the
$\varepsilon$-approximation of $Cs(x)$ in the
$\varepsilon$-neighborhood of the points $x\in M$, verifying the
continuity condition $h_\varepsilon (x)=Cs(x)$ on the boundary of
the $\varepsilon$-neighborhood \cite{DancaCodr}. In this way, the
discontinuous IVP became a continuous one of fractional order and a
numerical scheme for fractional-order differential equations, such
as the Adams-Bashforth-Moulton method presented in Subsection
\ref{CF}, can be used.

We consider for this case the fractional variant of a discontinuous
Chua system presented by \cite{Brown} (Table \ref{tab2}) for
$q=0.98$ (smaller values gives not rich dynamics).

\noindent As it can be seen from the bifurcation diagram
(Fig.\ref{fig15}), for $p\in(12, 12.55)$, there exists a narrow band
of a kind of ``chaotic saddle''. Within this window, the underlying
chaotic attractors look as being ``embedded'' within this transient
chaos (see for example the attractor $A_1$ in Fig.\ref{fig16})).

\begin{itemize}
\item Using the scheme $[2p_1,p_2]$ with $p_1=12.5$ and
$p_2=17$, the obtained synthesized attractor $A^\circ$ is \emph{AI}
with $A^*$ for $p^*=14$ (Fig.\ref{fig16}).
\end{itemize}

As shown in \cite{Danca6}, we found numerically that in these
systems of lower than third-order (i.e. $3*q$ which, for $q<1$, is
less than 3) chaos still may appear (as it is known, in the case of
integer order, according to the well-known Poincar\'{e}-Bendixon
theorem, chaos appears only at systems of minimum order three).

\section*{Conclusions}

In this review we have presented the parameter switching algorithm
according to which any attractor of a dynamical system belonging to
a large class of systems, may be numerically approximated
(synthesized). The attractors synthesis is achieved by using the
\emph{PS} algorithm, which switches periodically or randomly the
parameter. This facility is enabled by the convexity of $p$. The
average and synthesized attractors are AI and their underlying
trajectories, time dephased. The review is legitimated by the more
than ten published papers each of them containing several
applications.

As expected, the performance of the \emph{PS} algorithm is limited
due to the errors of the used numerical method, the length of the
time-subintervals $I_k$, $k=1,2,...,N$, the number of digits for
$p$, the step size $h$ and the distance in the parameter space
between $p_k$. Thus, we found that $N$ is not a critical parameter
(it could be even about $100$). The length of $I_k$ (i.e. the value
for $m_k$) is a critical parameter indicating for how long time the
control parameter of the considered system can take the values
$p=p_k$. We found that a maximum value for $m_k$ can be taken about
$25h$. For $p$, $3-4$ digits are enough to be compatible with the
smallest distance between the $p_k$ in the bifurcation diagrams.
Moreover, some real physical chaotic systems may have an infinite
number of different states or limit cycles with infinite period. But
a computer simulated system has a finite number of states; if the
precision of the computer is $n$ bits and the system to be modeled
has $k$ variables, the total number of system states is limited to
$2^{k*n}$; hence, given a determined state, it will be repeated
sooner or later and the system will become periodic, with a period
equal to the separation of the two occurrences of the state. The PS
method can alleviate this inaccuracy and make possible the
approximation of a computer simulated system to a real one, although
it may be necessary to use a sequence of parameter values lasting as
the whole segment of the system to be modeled (se e.g. \cite{Li}).

Some open problems are: the analytical proofs for the Lemma
\ref{lemma} for \emph{CF}, \emph{DI} and \emph{DF} systems, not only
for \emph{CI} as done in \cite{Yu}; an analytical proof for the
continuity of the bijection $F$; a detailed study of the time delay
between the trajectories of $A^\circ$ and $A^*$; the effect of noise
on the results; a comparison with the complex systems (fractals),
where the parameter switching may lead also to some interesting
results \cite{Almeida}. \vspace{10mm}

\noindent\textbf{APPENDIX}
\appendix
\section{Cross Correlation}\label{C}

As known, the cross-correlation of two signals is a measure of the
similarity of two waveforms. The cross-correlation has ranges from
-1.0 to +1.0. The closer it is to $+1$ or $-1$, the more closely the
two compared variables are related. The correlation of two signals
(the attractors underlying trajectories in our case) may indicate
that one of them is delayed in time with respect to the other. The
maximum value (close to unity) of this cross-correlation is obtained
when the two signals are in closest alignment with each other. The
value $-1$ means the signals are identically matched but opposite in
phase, while a value approaching zero indicates a low degree of
similarity (see the blue band around the horizontal axe in our
images). In this paper, the results were obtained with the
\emph{crosscorr} Matlab function with approximate $95$ percent
confidence interval.

\section{Hausdorff distance}\label{D}

The Hausdorff distance in a metric space, measures how far two
compact nonempty subsets are from each other. The classical
Hausdorff distance between two (finite) sets of points, \emph{A} and
\emph{B}, is defined as \cite[p.114]{Falconer}
\begin{equation*}
{D_H}(A,B) = \max \left\{ {\mathop {\sup }\limits_{x \in A} \mathop
{\inf }\limits_{y \in B} d(x,y),\mathop {\sup }\limits_{y \in B}
\mathop {\inf }\limits_{x \in A} d(x,y)} \right\},\
\end{equation*}
\noindent where $d(x,y)$ is the classical distance between two
points in the considered space.

\noindent If the two sets are curves, $D_H$ is defined as the
maximum distance to the closest point between the curves. Thus, if
the curves are defined as the sets of ordered pair of coordinates
$A=\{a_1,a_2,\ldots,a_n \}$ and $B=\{b_1,b_2,\ldots,b_m\}$ the
distance to the closest point between a point $a_i$ to the set
\emph{B} is \[d({a_i},B) = \mathop {\min }\limits_j \left\| {{b_j} -
{a_i}} \right\|.\] Thus, the Hausdorff distance is
\begin{equation*}
{d_H}(A,B) = \max \left\{ {\mathop {\max }\limits_i \left\{
{d({a_i},B)} \right\},\mathop {\max }\limits_j \left\{ {d({b_j},A)}
\right\}} \right\}.\
\end{equation*}

\section{Sketch of the analytical proof of Lemma \ref{prop}\label{A}}

Next, the main steps of the proof presented in \cite{Yu} for the
lemma, ensuring the \emph{AI} between $A^\circ$ and $A^*$ in the
case of \emph{CI} systems, are pointed out.

Consider the IVP (\ref{IVP}) with $C=O_{n\times n}$ and $q=1$
satisfying the assumptions stated in Section \ref{doi} and expressed
for the general case of $\mathbb{R}^n$, in the following form
\begin{equation} \label{eq1}
\dot{x}(t) = f(x(t)) + p\left(t/\lambda\right) B x(t), ~x(0) = x_0,
t \in I=[0,\infty),
\end{equation}
\noindent where $\lambda \in \mathbb{R}_{+}^{*}$ is a positive real
number which will be stated later, and $p:I\rightarrow \mathbb{R}^n
$ is considered as a piecewise continuous periodic function with
period $T_0$, and mean value $q$, i.e.
\begin{equation*}
\frac{1}{ T_0}\int_t^{t+ T_0} p(u)du = q, ~ t\in I. 
\end{equation*}
\noindent Let us define the average model of (\ref{eq1}) obtained
with the $\emph{PS}$ algorithm, expressed as follows
\begin{equation}
\dot{y} = f(y) + q B y,  ~ y(0) = y_0. \label{eq3}
\end{equation}

\noindent The IVP (\ref{eq1}) models the \emph{PS} algorithm and
generates the synthesized attractor $A^\circ$, while the IVP
(\ref{eq3}) represents the system whose solution approximates the
average attractor $A^*$.

\noindent We have to prove that the solutions of the equations
(\ref{eq1}) and (\ref{eq3}) differ by less than $\lambda^{2}$ for
$\lambda$ sufficiently small, via the so called {\it order function}
defined in terms of approximations\footnote{The order function
$\delta(\lambda^2)$, introduced in \cite[p.11]{San}, implies that
there exists $k$ s.t. $|\delta(\lambda^2)| \leq k \lambda^2$ when
$\lambda$ is sufficiently small.}.

\noindent Let next suppose that (\ref{eq3}) satisfies the assumption
\textbf{H1} and admits $s:I\rightarrow \mathbb{R}$ as the unique
solution.

\noindent Linearizing (\ref{eq3}) on a neighborhood of $s$, one
obtains the following IVP
\begin{equation}
\dot{\varepsilon}(t) = [E(t) + q B ] \varepsilon(t), \qquad
\varepsilon(0)=\varepsilon_0, \label{eq4}
\end{equation}
where $\varepsilon(t)= y(t) - s(t)$ and $E(t)$ denotes the Jacobian
of $f$ evaluated at $s(t)$.

\noindent Because $s(t)$ is the solution in (\ref{eq3}),
$\varepsilon(t)=0$ for $t \in I$, should be a solution of
(\ref{eq4}).

\noindent If we linearize the IVP (\ref{eq1}) for $x \in \Gamma_s$
(the domain of attraction of  $\varepsilon=0$) one obtains
\begin{equation*}
\dot{e}(t) = [E(t) + p(t/\lambda) B ] e(t), \qquad e(0)=e_0,
\end{equation*}
where $e(t) = x(t) - s(t)$.

\noindent Then, the theorem ensuring the \emph{AI} between the
attractors of the dynamical system modeled by the IVP (\ref{eq1})
and IVP (\ref{eq3}) can be enounced.

\mbox{}

\noindent\textbf{Theorem} Let assume that Eq.~(\ref{eq4}) is
uniformly exponentially stable, i.e. there exist the constants
$K>0,\mu>0$ such that
\begin{equation*}
\varepsilon(t) \leq K || \varepsilon_0 || \exp(-\mu t). 
\end{equation*}
Then, for $e_0 = \varepsilon_0$, there exists a positive scalar
$\lambda
> 0$, such that $\lim_{t \rightarrow \infty}||e(t)-\varepsilon(t)||
= \delta(\lambda^2)$, where $\delta(\lambda^2)$ is an order
function.

\begin{proof}
The complete proof can be found in \cite{Yu} and it mainly follows
the idea given in Chapter 4 of \cite{San}. The existence interval
$I$ is partitioned as follows: $I = \left[0, \lambda T\right]
\bigcup \left[\lambda T_0, 2\lambda T_0\right] \cdots$. In each
subinterval $I_{n} = \left[n\lambda T_0, (n+1)\lambda T_0 \right],
n=1,2,\ldots$, Eq.~(\ref{eq4}) has the solution $\varepsilon_n(t)$.
If on these subintervals, the initial condition is chosen
$\varepsilon_n(n\lambda T_0) = e(n\lambda T_0)$, using a generalized
Peano-Baker series \cite{Dacunha}, the Gronwall's inequality,
through straightforward algebraic operations, the following
inequality is inductive proved
\begin{equation*}
||e((n+1)\lambda T_0)-\varepsilon_n((n+1)\lambda T_0)|| \leq
\delta(\lambda^2),
\end{equation*}
for any $n$. Taking the limit $n \rightarrow \infty$, the proof is
complete.
\end{proof}

\section{Adams-Bashforth-Moulton scheme for fractional
ODEs} \label{B}

Next, a brief presentation of the Adams-Bashforth-Moulton scheme
\cite {Diethlem et al} is presented. Let consider the IVP
(\ref{IVPfrac}). Specifically, the method implies a discretization
of $I$ with grid points $t_{i}=hi, \ \ i=0,1,\ldots$ with a
preassigned step size \emph{h}. First, a preliminary approximation
$x_{i+1}^{\mathrm P}$ for $x(t_i)$ (the \emph{predictor phase}) is
computed via the formula
\begin{equation*}
  \label{eq:frac-AB}
  x_{i+1}^{\mathrm P} = \sum_{j=0}^{\lceil q \rceil-1}
      \frac{t_{i+1}^j}{j!} x_0^{(j)}
  + \frac{1}{\Gamma(q)} \sum_{j=0}^{i} b_{j,i+1} g(x_{j}),
\end{equation*}
\noindent where $b_{j,i+1}$ have the form
\begin{equation*}\label{eq:weights-uniform-rectangle}
  b_{j,i+1} = \frac{h^q}q \left( (i+1-j)^q - (i-j)^q \right).
\end{equation*}

\noindent Then, the final approximation $x_{i+1}$ for $x(t_{i+1})$
(the \emph{corrector phase}) is
\begin{align*}\label{eq:frac-AM}
x_{i+1} =& \sum_{j=0}^{\lceil q\rceil-1}\frac{t_{i+1}^j}{j!}
x_0^{(j)}+ \frac{h^q}{\Gamma(q+2)}
            \left( \sum_{j=0}^{i} a_{j,i+1} g(x_{j}) +
                    g(x_{i+1}^{\mathrm P}) \right),
\end{align*}
\noindent with
\begin{equation*}
\label{eq:weights-uniform-trap1} a_{0,i+1} =
   i^{q+1} - (i-q) (i+1)^q,
\end{equation*}
\noindent and
\begin{equation*}
\label{eq:weights-uniform-trap2} a_{j,i+1} =
  (i-j+2)^{q+1} + (i-j)^{q+1} - 2 (i-j+1)^{q+1},
\end{equation*}
\noindent for $j = 1, 2,\ldots, i$.

The Gamma function, $\Gamma$, is approximated in this work with the
so-called Lanczos approximation \cite{Press}
\begin{equation*}
\Gamma(z)=\frac{\sum_{i=0}^{6} p_{i}z^{i}} {\prod_{i=0}^{6}(z+i)}
(z+5.5)^{z+0.5}{\mathrm e}^{-(z+5.5)}.
\end{equation*}

\noindent for $z \in \mathbb C$ with $\mathop{\mathrm{Re}}(z)>0$.
The coefficients $p_i$ are shown in Table \ref{tab4}.

While in the standard methods for ODEs of integer order, the current
approximation $x_{k}$ depends only on the results of a few backward
steps, like all reasonable numerical methods for fractional
differential equations, the fractional scheme
Adams-Bashforth-Moulton requires the entire backward integration
history at each point in time. Thus, each calculated value $x_{k}$
depends on all previous values $x_0,x_1,\ldots,x_{k-1}$. This
characteristic implies a serious drawback with respect to the
required computing time. For example, to obtain $4000$ points within
some attractor, about $8\times 10^6$ iterations are necessary.
However, this is necessary to appropriately reflect the memory
effects possessed by fractional differential operators.

\noindent A detailed analysis of this method can be found in
\cite{Diethelem2} and a background on fractional differential
equations is presented in \cite{kaicarte}.

\section{Filippov regularization}\label{filippov}
Let consider the following general IVP with discontinuous right-hand
side
\begin{equation}
\dot x = f(x),\,\,\,\,x(0) = {x_0},\,\,\,t \in I, \label{disc}
\end{equation}
\noindent with $f$ locally bounded on $\mathbb{R}^n$. In particular,
the discontinuity is due to the discontinuity of the state variable,
of the associated vector field, of Jacobian (partial derivatives) or
higher order discontinuity. \noindent The continuity domain consists
in a finite $m$ number of open regions ${D_i} \subset
{\mathbb{R}^n},\,\,\,i = 1,2, \ldots ,m$, the discontinuity set $M$
being $M$ = ${\mathbb{R}^n}\backslash \bigcup\nolimits_{i = 1}^m
{{D_i}}$.

\noindent The IVP (\ref{disc}) may have not any solutions in the
classical sense. Therefore, in this paper we have chosen the way
given by \cite{Filippov}, which transforms the IVP (\ref{disc}), via
the differential inclusion approach, into a multi-valued Cauchy IVP
\begin{equation}\label{di}
\dot x \in F(x),\,\,\,\,x(0) = {x_0},\,\,\,\,for~ almost ~all~ t \in
I,
\end{equation}
\noindent where, $F:\mathbb{R}^n\rightrightarrows \mathbb{R}^n$ is a
set-valued function into the set of all subsets of $\mathbb{R}^n$.
For our class of systems, \emph{F} can be defined using the so
called \emph{Filippov regularization}
\begin{equation*}
F(x) = con\mathop {\lim }\limits_{x' \to x} f(x'),
\end{equation*}
\noindent where $con$ means the convex hull and $\lim _{x' \to
x}f(x')$ represents the set of all limits for all convergent
sequences $f(x_k)$ with $x_k\rightarrow x$. For $x\in M$, $F(x)$ is
a set, while for $x\notin M$, $F(x)$ consists in a single point
$f(x)$. As example, for the sign function the Filippov
regularization gives the following set-valued function
\begin{equation*}
\ Sgn(x) = \left\{ {\begin{array}{*{20}{c}}
   {\{  - 1\} ,\,\,x < 0,}  \\
   {[-1,1]\,\,\,\,x = 0,}  \\
   {\{  + 1\} \,\,\,\,\,\,x > 0.}  \\
\end{array}} \right.\
\end{equation*}
For example, the Filippov regularization applied to the Example
(\ref{exdisc}) leads to the following differential inclusion
\begin{equation} \label{exsetval}
\dot x \in 2 - 3{\mathop{\rm Sgn}} (x).
\end{equation}
\noindent\textbf{Definition} \cite{Filippov} A \emph{generalized
solution} (or \emph{Filippov solutions}) of the IVP (\ref{disc}) is
an absolutely continuous vector-valued function $x:I \to {R^n}$
verifying the IVP (\ref{di}) for a.a. $t \in I$.

\mbox{}

\noindent Even the IVP (\ref{disc}) may have no classical solutions,
the setvalued IVP (\ref{di}) may have a unique or several
generalized solutions. For example, the equation (\ref{exdisc}),
after regularization becomes the set-valued problem (\ref{exsetval})
and has the following generalized solutions: if $x_0>0$, then
$x(t)=-t+x_0$ for $t<x_0$ and $x(t)=0$ for $t\geq x_0$. In other
words, the solution can be prolonged continuous along the axis
$x=0$. If $x'_0<0$, the solution is $x(t)=5t+x'_0$ for $t<x'_0$ and
$x(t)=0$ for $t\geq x'_0$ (Fig.\ref{fig11})

The background of differential inclusions and their solutions can be
found e.g. in \cite{Aubin1} and \cite{Aubin2}. The existence and
uniqueness of solutions for our class of \emph{DI} systems are
presented in \cite{Dancaunguri} and will be not considered here.

To solve numerically the IVP (\ref{disc}) special numerical methods
for differential inclusions are necessary. However, for our class of
IVPs, due to the presence of $s$ functions, the discontinuity
appears only in a finite null set \emph{M}, where actually the IVP
is a set-valued problem. For the points $x \in {D_i}$, the IVP is a
continuous problem. Therefore, we can integrate in $D_i$ the IVP
(\ref{disc}) using e.g. the standard Runge-Kutta method, while for
$x \in M$, a numerical method for the corresponding differential
inclusion $\dot x \in F(x)$ has to be used. Precisely, when the
trajectory enters the discontinuity surface, we have to choose for
derivative of solution, generally for some finite time, a value
within the set $F(x)$ while a numerical method for differential
inclusions is used (the simplest one is the adapted forward Euler
method see e.g. \cite{Lempio2,Lempio}). For example, when $x=0$ in
the Example (\ref{exdisc}), we have to solve the differential
inclusion $\dot x \in [ - 1,5]$. There are several possibilities to
manage this problem using e.g. so called \emph{selection strategies}
(see \cite{Kastner}). In this paper we used the simplest way, namely
the \emph{random strategy} which implies a randomly choice of a
value within $F(x)$ (the interval $[- 1,5]$ in our example). There
are several possibilities to find the moments when the trajectory
enters and leaves the discontinuity surfaces (see e.g.
\cite{Wierci}) during which the chosen method solves the
differential inclusion.


\clearpage

\begin{table}[t]
\caption{Classification of the considered dynamical systems modeled
by the IVP (\ref{IVP}).} \label{tab1}
\begin{tabular}{c|p{75mm}p{75mm}}     
\hline \noalign{\smallskip}
{} &{} $C = {0_{n\times n}}$ & $C\ne {0_{n\times n}}$\\
\noalign{\smallskip}\hline\noalign{\smallskip}
$q = 1$ & Continuous of Integer order (CI) systems &Discontinuous of Integer order (DI) systems\\
$q \in (0,1)$ & Continuous of Fractional order (CF) systems &
Discontinuous of Fractional order (DF) systems\\[6pt]
\hline
\end{tabular}
\end{table}

\begin{table}[t]\footnotesize
\caption{Systems utilized in this paper.} \label{tab2}
\begin{tabular}[c]{@{}cclcccc@{}}
\hline\noalign{\smallskip}
Type & Order & System  & $q$ & $f(x)$ & $B$ & $C$ \\
\noalign{\smallskip}\hline
\begin{tabular}[c]{c} \\
\rotatebox{90}{Continuous }
\end{tabular}
&
\begin{tabular}[c]{c}\\
\rotatebox{90}{Integer}\\
\end{tabular}
&
\begin{tabular}[c]{l}%
Lorenz \\
$\dot{x}_{1}=10(x_{2}-x_{1})$\\
$\dot{x}_{2}=-x_{1}x_{3}-x_{2}+px_{1}$  \\
$\dot{x}_{3}=x_{1}x_{2}-\frac{8}{3}x_{3}$
\end{tabular}
& $1$ & $\left(
\begin{array}[c]{c}%
{10(x}_{2}{-x}_{1}{)}\\
{-x}_{1}{x}_{3}{-x}_{2}\\
{x}_{1}{x}_{2}{-}\frac{8}{3}{x}_{3}%
\end{array}
\right)$ & $\left(
\begin{tabular}
[c]{lll}%
$0$ & $0$ & $0$\\
$1$ & $0$ & $0$\\
$0$ & $0$ & $0$%
\end{tabular}
\right)$ &
$O_{3\times3}$\\
\hline
\begin{tabular}[c]{c}%
\\
\rotatebox{90}{Continuous}\\
\\
\end{tabular}
&
\begin{tabular}
[c]{c}%
\\
\rotatebox{90}{Fractional}\\
\\
\end{tabular}
& 
\begin{tabular}[c]{l}%
L\"{u} \\

$\frac{d^{q}x_{1}}{dt^{q}}=-x_{1}+px_{2}$\\
\noalign{\smallskip}
$\frac{d^{q}x_{2}}{dt^{q}}=-x_{1}x_{3}+28x_{2}$\\
\noalign{\smallskip}
$\frac{d^{q}x_{3}}{dt^{q}}=x_{1}x_{2}-3x_{3}$%
\end{tabular}
& $q<1$ & $\left(
\begin{array}
[c]{c}%
{\-x}_{1}   \\
{-x}_{1}{x}_{3}{+28x}_{2}\\
{x}_{1}{x}_{2}{-3x}_{3}%
\end{array}
\right)$ & $\left(
\begin{tabular}
[c]{lll}%
$0$ & $1$ & $0$\\
$0$ & $0$ & $0$\\
$0$ & $0$ & $0$%
\end{tabular}
\right)  $ & $O_{3\times3}$\\
\hline
\begin{tabular}[c]{c}%
\\
\rotatebox{90}{Discontinuous}\\
\end{tabular}
&
\begin{tabular}[c]{c}%
\\
\rotatebox{90}{\mbox{Integer}}\\
\end{tabular}
& 
\begin{tabular}[c]{l}%
Sprott \\
$\dot{x}_{1}=x_{2}$\\
$\dot{x}_{2}=x_{3}$\\
$\dot{x}_{3}=-x_{1}-x_{2}-px_{3}+sgn(x_{1})$%
\end{tabular}
& $1$ & $\left(
\begin{array}[c]{c}%
{x}_{2}\\
{x}_{3}\\
{-x}_{1}{-x}_{2}%
\end{array}
\right)$ & $\left(
\begin{tabular}[c]{lll}%
$0$ & $0$ & $0$\\
$0$ & $0$ & $0$\\
$0$ & $0$ & $-1$%
\end{tabular}
\right)$ & $\left(
\begin{tabular}[c]{lll}%
$0$ & $0$ & $0$\\
$0$ & $0$ & $0$\\
$1$ & $0$ & $0$%
\end{tabular}
\right)$\\
\hline
\begin{tabular}[c]{c}%
\\
\rotatebox{90}{Discontinuous} \\
\\
\end{tabular}
&
\begin{tabular}[c]{c}%
\\
\rotatebox{90}{Fractional} \\
\\
\end{tabular}
&
\begin{tabular}[c]{l}%
Chua\\
$\frac{d^{q}x_{1}}{dt^{q}}=-2.57x_{1}+9x_{2}+3.86sgn(x_{1})$\\
\noalign{\smallskip}
$\frac{d^{q}x_{2}}{dt^{q}}=x_{1}-x_{2}+x_{3}$\\
\noalign{\smallskip}
$\frac{d^{q}x_{3}}{dt^{q}}=-px_{2}$%
\end{tabular}
& $q<1$ & $\left(
\begin{array}[c]{c}%
{-2.57x}_{1}{+9x}_{2}\\
{-x}_{1}{-x}_{2}{+x}_{3}\\
{ 0}%
\end{array}
\right)$ & $\left(
\begin{tabular}[c]{lll}%
$0$ & $0$ & $0$\\
$0$ & $0$ & $0$\\
$0$ & $-1$ & $0$%
\end{tabular}
\right)$ & $\left(
\begin{tabular}[c]{lll}%
$3.86$ & $0$ & $0$\\
$0$ & $0$ & $0$\\
$0$ & $0$ & $0$
\end{tabular}
\right)$\\
\hline
\end{tabular}
\end{table}

\begin {table}[ht]
\caption{Pseudo-code of the $\emph{PS}$ algorithm} \label{tab3}
\begin{center}
\begin{tabular}[c]{l}%
\hline\noalign{\smallskip}
CHOOSE $\emph{$S_h$, T, h}$\\
REPEAT \\
\ \ \ \ \ \ \ FOR $i=1$ to $N$ \\
\ \ \ \ \ \ \ \ \ \ \ \ \ \ FOR $k=1$ to $m_{i}$ \\
\ \ \ \ \ \ \ \ \ \ \ \  \ \ \ \ \ \ \ \ \ one step
integration of the IVP (\ref{IVPsimplu}) for $p=p_i$\\
$\ \ \ \ \ \ \ \ \ \ \ \ \ \ \ \ \ \ \ \ \ t=t+h$\\
\ \ \ \ \ \ \ \ \ \ \ \ \ \ ENDFOR\\
\ \ \ \ \ \ \ ENDFOR\\
UNTIL $t\geq T$ \\
\noalign{\smallskip}\hline
\end{tabular}
\end{center}
\end{table}

\begin{table}[h]
\caption{Coefficients of the Lanczos approximation.} \label{tab4}
\begin{center}
\begin{tabular}{lr}
$i$ & \multicolumn{1}{c}{$p_i$} \\
\noalign{\smallskip}\hline\noalign{\smallskip}
$0$ & ${ 75122.6331530}$ \\
$1$ & ${ 80916.6278952}$ \\
$2$ & ${ 36308.2951477}$ \\
$3$ & ${ 8687.2452971} $ \\
$4$ & ${ 1168.9264948} $ \\
$5$ & ${ 83.8676043}   $ \\
$6$ & ${ 2.5066283}    $ \\
\end{tabular}
\end{center}
\end{table}

\clearpage

\begin{figure}[!h]
\includegraphics[scale=0.3]{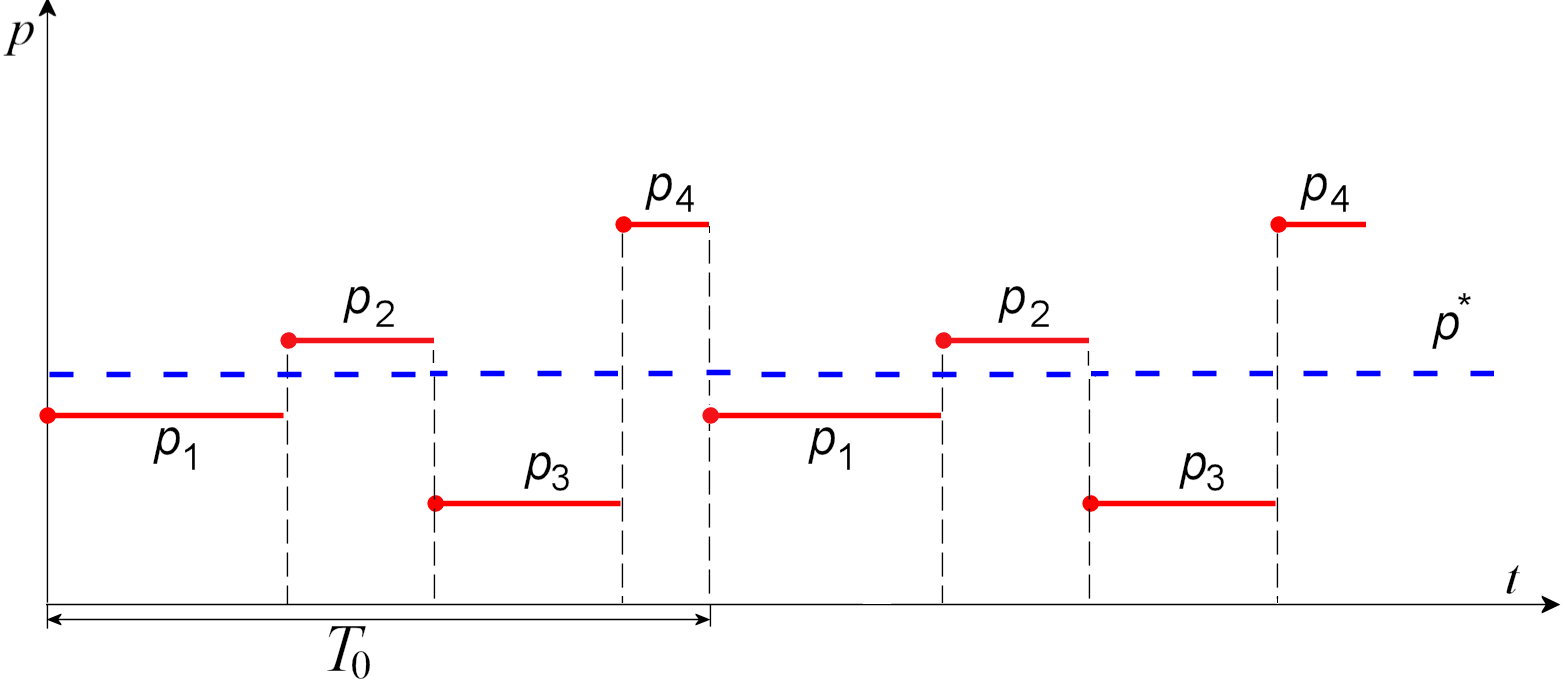}
\caption{Piecewise constant periodic function $p:I\rightarrow
\mathbb{R}$ (sketch).} \label{fig0}
\end{figure}

\begin{figure*}[t]
\begin{center}
  \includegraphics[width=0.9\textwidth]{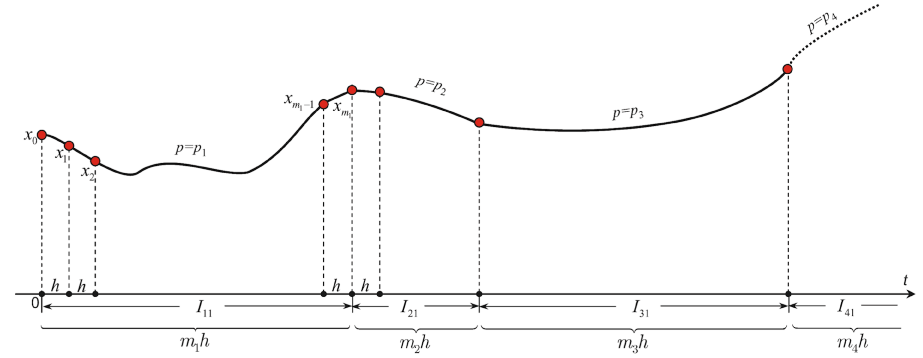}\\
\caption{Time subintervals $I_{i1}, i=1,2,3$ and $4$ (sketch)}
\label{fig1}
\end{center}
\end{figure*}

\begin{figure}[!b]
\begin{center}
\includegraphics[clip,width=0.5\textwidth]{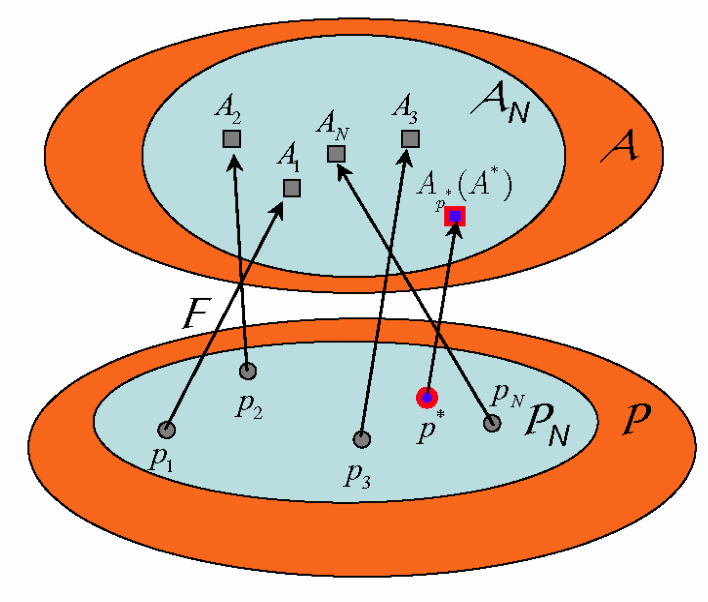}
\end{center}
\caption{Bijection $F:\mathcal{P}_{N}\rightarrow\mathcal{A}_{N}$
(sketch).} \label{fig2}
\end{figure}

\begin{figure}[t]
\begin{center}
\includegraphics[clip,width=0.6\textwidth]{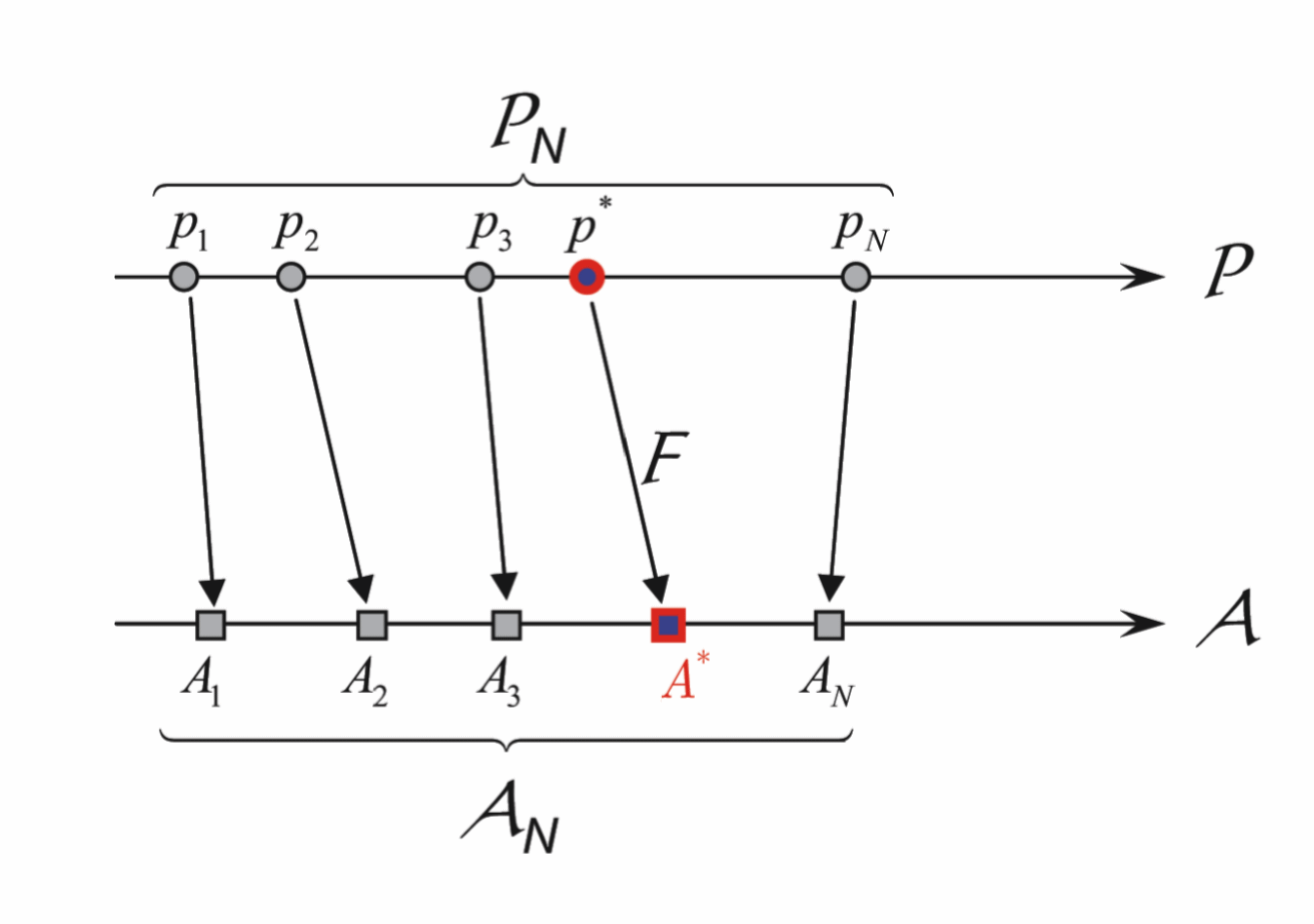}
\end{center}
\caption{Order induced by $F$ in $\mathcal{A}_{N}$.} \label{fig3}
\end{figure}

\begin{figure}[t]
\begin{center}
\includegraphics[clip,width=0.5\textwidth]{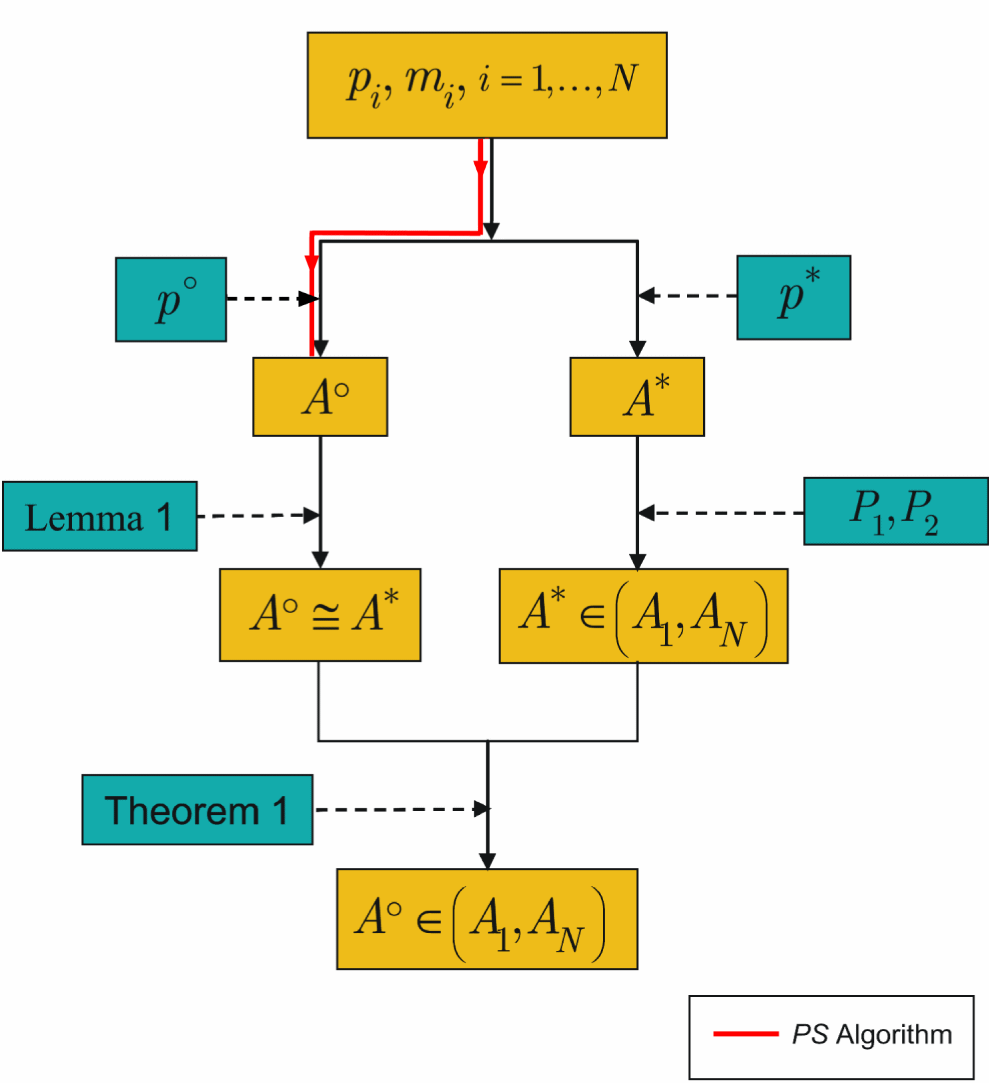}
\end{center}
\caption{Attractor synthesis: sketch of the proof of Theorem
\ref{prop}.}\label{fig00}
\end{figure}

\begin{figure*}[ht]
\begin{center}
\includegraphics[clip,width=0.8\textwidth]{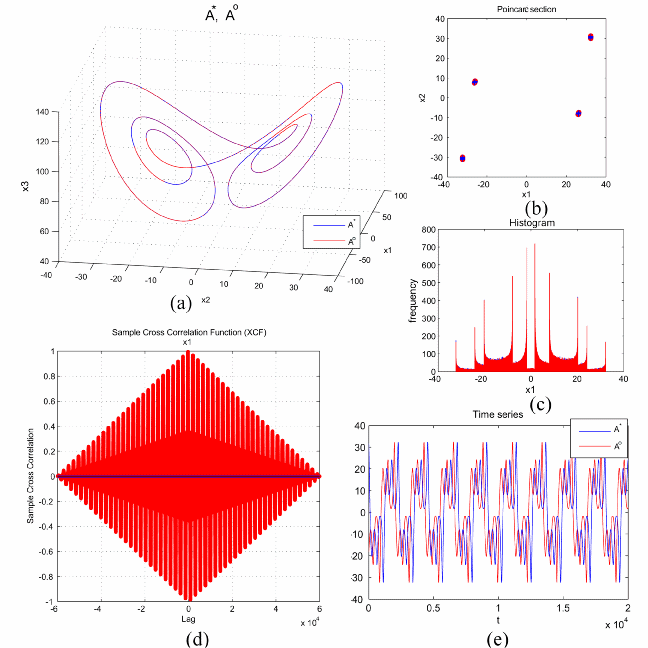}
\end{center}
\caption{Scheme $[{m_1}{p_1},{m_2}{p_2}]$ with $m_1=90$, $m_2=96$,
and $p_1=p_2=1$ applied to the Lorenz system. $p^*=93$. (a)
$A^\circ$ and $A^*$; (b) Poincar\'{e} sections; (c) Histograms; (d)
Cross-correlations; (e) Time series.\vspace{98mm}} \label{fig4}
\end{figure*}

\begin{figure*}[!]
\begin{center}
\includegraphics[clip,width=0.8\textwidth]{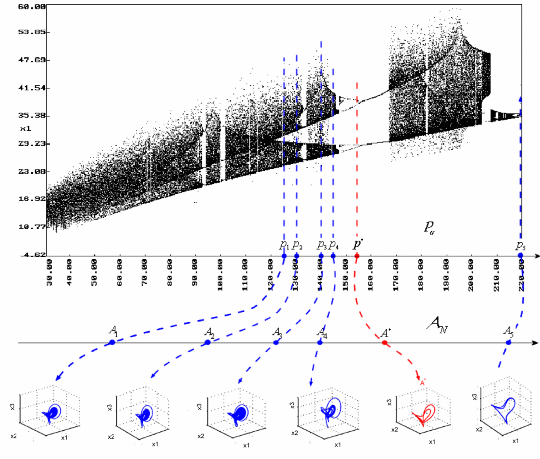}
\end{center}
\caption{Bifurcation diagram for the Lorenz system.\vspace{107mm}}
\label{fig5}
\end{figure*}

\begin{figure*}[!t]
\begin{center}
\includegraphics[clip,width=0.8\textwidth]{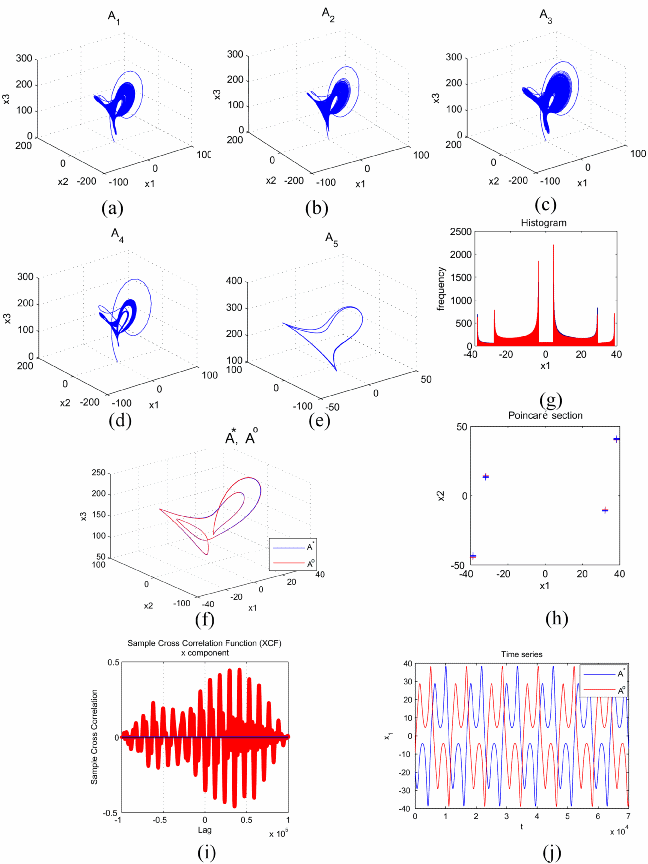}
\end{center}
\caption{Scheme $\left[{2{p_1},3{p_2},2{p_3},4{p_4},3{p_5}} \right]$
for $p_1=125$, $p_2=130$, $p_3=140$, $p_4=144$ and $p_5=220$ applied
to the Lorenz system. $p^*=154$. (a)-(e) The attractors $A_i$
corresponding to $p_i, i=1,\ldots,5$; (f) $A^\circ$ and $A^*$; (g)
Histograms; (h) Poincar\'{e} sections; (i) Cross-correlations; (j)
Time Series.\vspace{125mm}} \label{fig6}
\end{figure*}

\begin{figure}[t]
\begin{center}
\includegraphics[clip,width=0.8\textwidth]{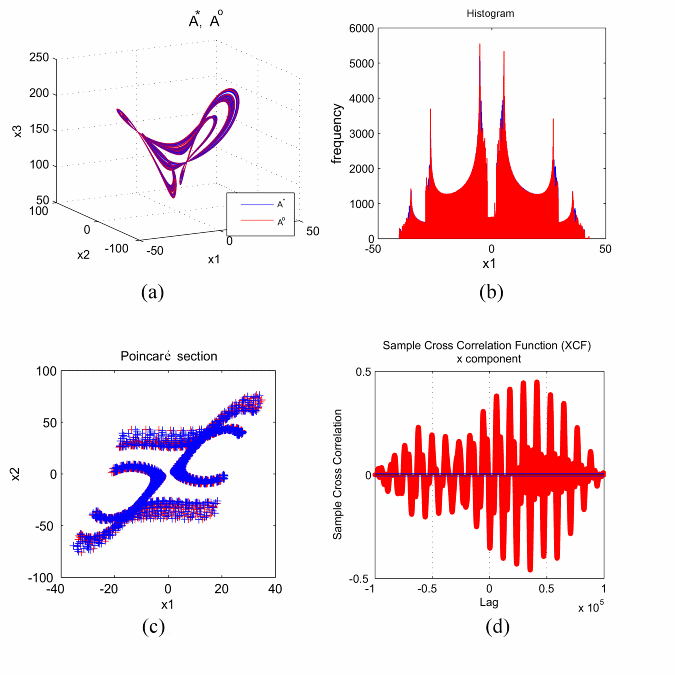}
\end{center}
\caption{Same scheme as in Fig.\ref{fig6}: $\left[
{2{p_1},3{p_2},2{p_3},4{p_4},3{p_5}} \right]$ but with $p_5=166$
instead $p_5=220$. $p^*=142.428$. (a) $A^\circ$ and $A^*$; (b)
Histograms; (c) Poincar\'{e} sections; (d) Cross-correlations.}
\label{fig7}
\end{figure}

\begin{figure}[t]
\includegraphics[clip,width=0.8\textwidth]{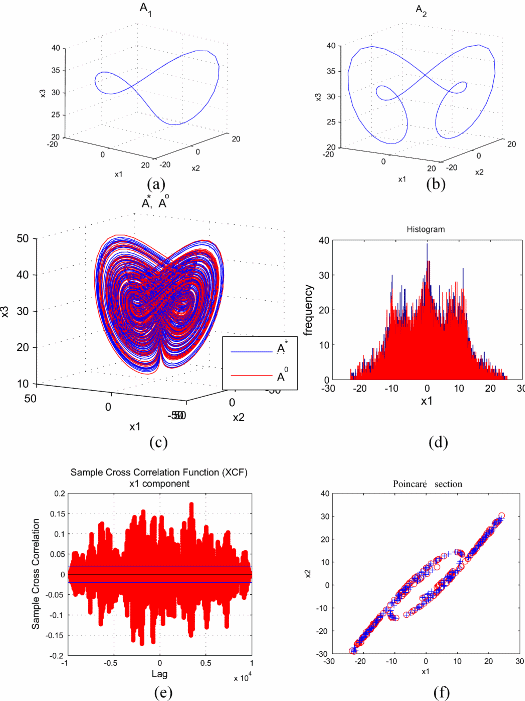}
\caption{Scheme $[1p_1,1p_2]$ with $p_1=32$ and $p_2=34.5$ applied
to the fractional L\"{u} system. $p^*=33.25$. (a),(b) The
attractors, $A_1$ and $A_2$; (c) The attractors $A^\circ$ and $A^*$;
(d) Histograms; (e) Cross-correlations; (f) Poincar\'{e} sections.}
\label{fig8}
\end{figure}

\begin{figure*}
\begin{center}
\includegraphics[clip,width=0.8\textwidth]{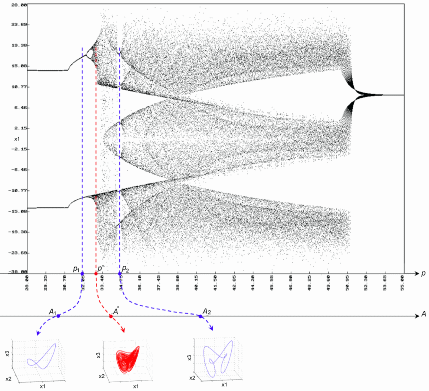}
\end{center}
\caption{Bifurcation diagram for the L\"{u} system.} \label{fig9}
\end{figure*}

\begin{figure}
\begin{center}
\includegraphics[clip,width=0.8\textwidth]{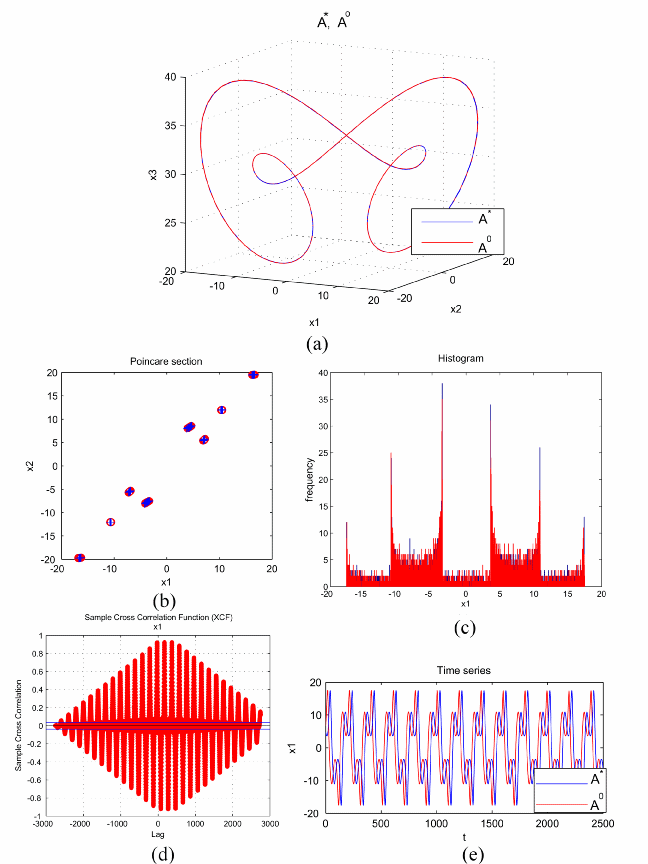}
\end{center}
\caption{Scheme $[1p_1,1p_2]$, with $p_1=33.5$ and $p_2=35.5$ for
the L\"{u} system. $p^*=34.5$. (a) $A^\circ$ and $A^*$; (b)
Poincar\'{e} sections; (c) Histograms; (d) Cross-correlations; (e)
Time series.} \label{fig10}
\end{figure}

\begin{figure}[t]
\begin{center}
\includegraphics[clip,width=0.6\linewidth]{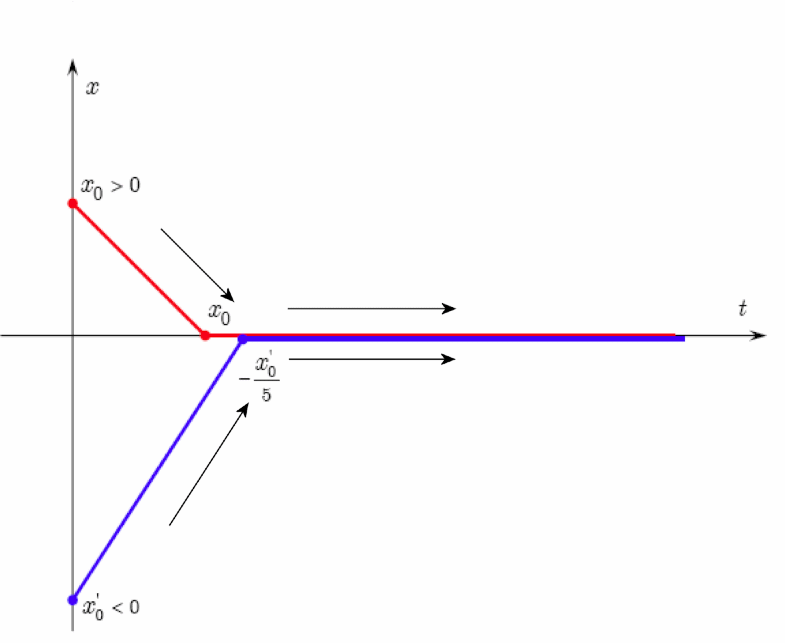}
\end{center}
\caption{Generalized solutions of the equation (\ref{exdisc})
(Sketch).} \label{fig11}
\end{figure}

\begin{figure*}
\begin{center}
\includegraphics[clip,width=0.8\textwidth]{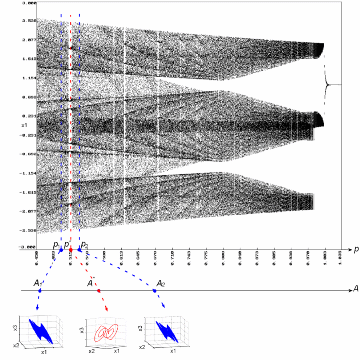}
\end{center}
\caption{Bifurcation diagram for the Sprott system.} \label{fig12}
\end{figure*}

\begin{figure*}
\begin{center}
\includegraphics[clip,width=0.8\textwidth]{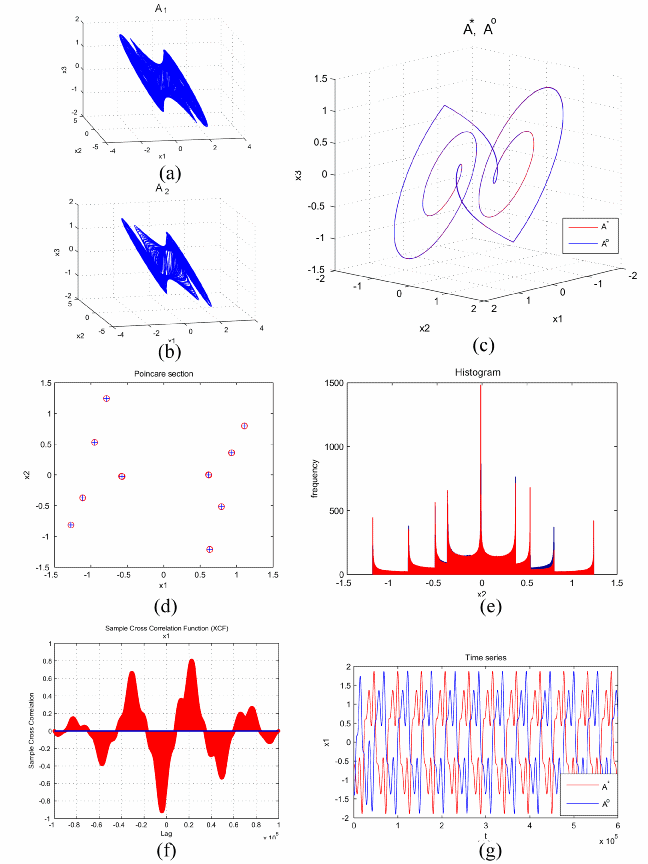}
\end{center}
\caption{Scheme $[1p_1,1p_2]$ for $p_1=0.5$ and $p_2=0.528$ applied
to the Sprott system. $p^*=0.514$. (a),(b) $A_1$ and $A_2$; (c)
$A^\circ$ and $A^*$; (d) Poincar\'{e} sections; (e) Histograms; (f)
Cross-correlations; (g) Time series.}  \label{fig13}
\end{figure*}

\begin{figure*}[!t]
\begin{center}
\includegraphics[clip,width=0.8\textwidth]{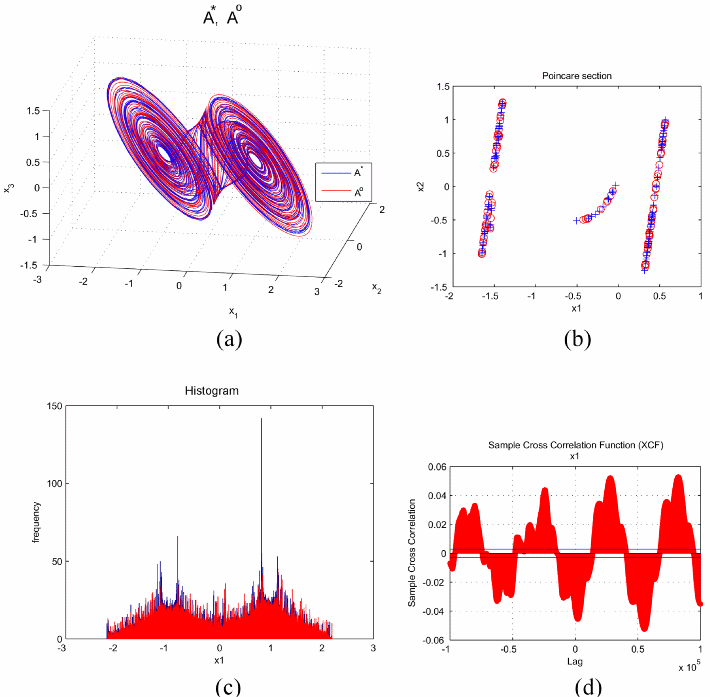}
\end{center}
\caption{Chaotic \emph{PS} algorithm applied to Sprott system for
$N=100$ and random with uniform distribution choice for $m_i$: $m_i
\in {{1,2,3}}$ and $p_i \in [0.45,0.65]$. (a) $A^\circ$ and $A^*$;
(b) Poincar\'{e} sections; (c) Histograms; (d) Cross-correlation.}
\label{fig14}
\end{figure*}

\begin{figure*}[!t]
\begin{center}
\includegraphics[clip,width=0.8\textwidth]{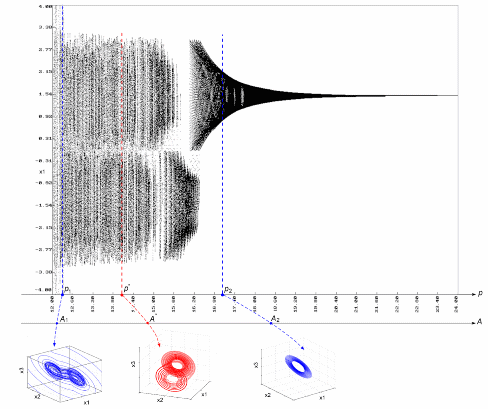}
\end{center}
\caption{Bifurcation diagram for the fractional Chua system.}
\label{fig15}
\end{figure*}

\begin{figure*}[!t]
\begin{center}
\includegraphics[clip,width=0.8\textwidth]{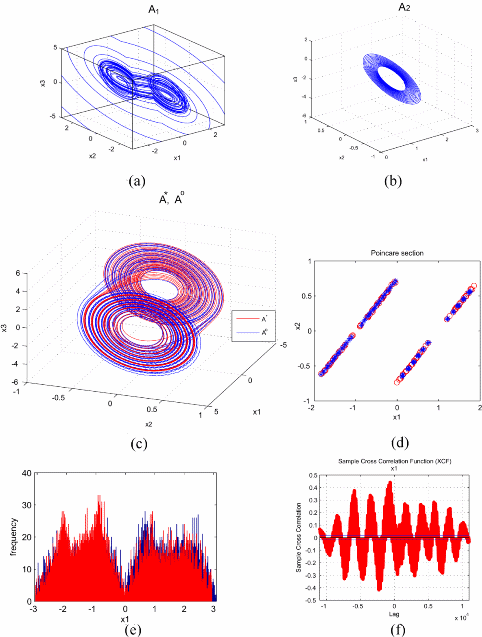}
\end{center}
\caption{Scheme $[2p_1,p_2]$ with $p_1=12.5$ and $p_2=17$, applied
to the fractional Chua system. $p^*=14$. (a),(b) $A_1$ and $A_2$;
(c) $A^\circ$ and $A^*$; (d) Poincar\'{e} sections; (e) Histograms;
(f) Cross-correlations.} \label{fig16}
\end{figure*}


\end{document}